\documentclass[11pt,a4paper]{article}

\usepackage[a4paper, margin=1in]{geometry}
\usepackage[T1]{fontenc}
\usepackage{amsmath,amssymb,amsthm}
\usepackage{booktabs}
\usepackage{graphicx}
\usepackage{algorithm}
\usepackage{algpseudocode}
\usepackage{hyperref}
\usepackage{cleveref}
\usepackage[expansion=false]{microtype}
\usepackage{xcolor}
\usepackage{enumitem}
\usepackage{caption}
\usepackage{subcaption}
\usepackage{array}
\usepackage{tabularx}
\usepackage{multirow}
\usepackage{url}

\newtheorem{definition}{Definition}[section]
\newtheorem{theorem}{Theorem}[section]
\newtheorem{corollary}{Corollary}[theorem]
\newtheorem{claim}{Claim}[section]
\newtheorem{remark}{Remark}

\hypersetup{
  colorlinks=true,
  linkcolor=blue!60!black,
  citecolor=green!50!black,
  urlcolor=blue!70!black,
  pdftitle={The Bureaucracy of Speed},
  pdfauthor={Vlad Novytskyy},
  pdfsubject={Multi-Agent Authorization Revocation},
  pdfkeywords={cache coherence, multi-agent systems, authorization, revocation, MESI, release consistency}
}

\newcommand{\CCS}{\textsc{CCS}}
\newcommand{\MESI}{\textsc{MESI}}
\newcommand{\RCC}{\textsc{RCC}}

\newcommand{\dphi}{\varphi}
\newcommand{\E}{\mathcal{E}}
\newcommand{\B}{\mathcal{B}}
\newcommand{\M}{\mathcal{M}}

\begin{document}

\title{\textbf{The Bureaucracy of Speed: Structural Equivalence Between\\
Memory Consistency Models and Multi-Agent Authorization Revocation}}

\author{
  Vladyslav Parakhin\\
  \textit{Senior Data Engineer}\\
  \textit{Okta}
}

\date{}
\maketitle

\begin{abstract}
The temporal assumptions underpinning conventional Identity and Access Management
collapse under agentic execution regimes.
A sixty-second revocation window---operationally negligible for a human
operator---permits on the order of $6 \times 10^3$ unauthorized API calls at
100~ops/tick; at AWS Lambda scale, the figure approaches $6 \times 10^5$.
This is a coherence problem, not merely a latency problem, and the authorization
revocation literature has largely failed to recognize it as such.

I define a \textit{Capability Coherence System} (\CCS{})---a tuple
$\langle A, C, \Sigma, \delta, \alpha, \B \rangle$---and construct a
state-mapping function $\dphi: \Sigma_{\MESI{}} \to \Sigma_{\mathrm{auth}}$
that preserves transition structure under bounded-staleness semantics.
A safety theorem (\cref{thm:rcc}) bounds unauthorized operations for the
execution-count Release Consistency-directed Coherence (\RCC{}) strategy at
$D_{\mathrm{rcc}} \leq n$, independent of agent velocity~$v$---a qualitative
departure from the $O(v \cdot \mathrm{TTL})$ scaling of time-bounded strategies.

Evaluation proceeds through tick-based discrete event simulation: three
business-contextualised scenarios, four strategies, ten deterministic seeds per
configuration (population~$\sigma$).
Eager invalidation yields $500.0 \pm 0$ unauthorised operations at
$v = 100$, $\Delta_{\mathrm{network}} = 5$; lease-based \textsc{TTL} yields
$6{,}000.0 \pm 0$; lazy check-on-use yields $2{,}400.0 \pm 0$;
\RCC{} at $n = 50$ yields $50.0 \pm 0$ ($\sigma = 0$, confirming
\cref{thm:rcc} exactly).
Under anomaly-triggered revocation, \RCC{} achieves a $184\times$ reduction
against lease \textsc{TTL} ($16.0 \pm 3.7$ vs.\ $2{,}950.8 \pm 3.6$).

\smallskip\noindent
\textbf{Contributions:}
(1)~a formal state-mapping between cache coherence and authorization revocation
semantics;
(2)~a Velocity Vulnerability metric $V_v = v \cdot \mathrm{TTL}$ establishing
agent velocity as a first-class security dimension;
(3)~an operation-bounded credential model grounded in release consistency that
enforces coherence at synchronisation boundaries;
(4)~reproducible multi-run evaluation with published source code at \url{https://github.com/hipvlady/prizm}.
\end{abstract}

\smallskip\noindent
\textbf{Subjects:} Multi-Agent Systems (cs.MA); Cryptography and Security (cs.CR);
Distributed, Parallel, and Cluster Computing (cs.DC)

\newpage
\tableofcontents
\newpage

\section{The Coherence Framing}
\label{sec:intro}

The agent has already executed 47 unauthorised API calls.
The authority revoked its credential at $t_r = 0$.
The \textsc{TTL} expires at $t = 60$.
Nothing in the current authorisation stack detects this.

A compromised agentic system operating under \textsc{TTL}-based credential
management accumulates $V_v = v \cdot \mathrm{TTL}$ unauthorised operations
before self-invalidation---a quantity linear in both velocity and window length.
At $v = 10{,}000$~TPS (AWS Lambda scale) and a 60-second \textsc{TTL},
$V_v = 6 \times 10^5$.
The Replit incident~\cite{oso2025}---an AI coding agent that deleted a production
database---is a concrete instance of this failure class, not an edge case.
The damage bound is not incidental to the authorisation architecture; it is a
structural property of the coherence regime the system operates under.

\textsc{IAM} protocols developed for biological operators inherit implicit
assumptions: session timeouts measured in minutes, revocation windows tolerable
at human interaction rates ($\sim 1$~req/s $\to 60$ unauthorised ops), and
eventual consistency acceptable as a design trade-off.
Autonomous agents violate every one of these assumptions simultaneously.
OAuth~2.0, OIDC, and their derivatives were not designed for entities capable of
recursive delegation, high-frequency decision loops, and parallel execution across
thousands of concurrent sessions.

The OpenID Foundation~\cite{openid2025} identifies revocation across
offline-attenuated delegation chains as ``largely unsolved.''
Chan et al.~\cite{chan2024} argue that AI systems require distinct, verifiable
identities rather than repurposed human credentials.
Nagabhushanaradhya~\cite{nagabhushanaradhya2025} proposes \textsc{OIDC-A}, an
OpenID Connect extension for agent identity.
Mei et al.~\cite{mei2024} treat \textsc{LLM} agents as schedulable processes
under an \textsc{AIOS} abstraction.
I extend this reasoning to its logical conclusion: if agents are processes, they
require \textit{coherence protocols} at the authorisation layer---not merely
authentication, and not merely rate limiting.

The central argument is that authorisation revocation in multi-agent delegation
chains is operationally equivalent to cache coherence in shared-memory
multiprocessors under bounded-staleness semantics.
I formalise this equivalence, derive bounded-staleness guarantees for four
revocation strategies, and evaluate them through simulation.

\paragraph{Contributions.}
\begin{enumerate}[leftmargin=*, label=\textbf{\arabic*.}]
\item \textbf{Formal equivalence.} A Capability Coherence System (\CCS{}) is
  defined, with state-mapping function $\dphi$ from \MESI{} states to
  authorisation states, preserving transition structure under bounded-staleness
  semantics (\cref{sec:formal}).

\item \textbf{Velocity Vulnerability metric.} Damage potential of \textsc{TTL}-based
  approaches is formalised as $V_v = v \cdot \mathrm{TTL}$, with proof that this
  bound is velocity-dependent while operation-count bounds are
  velocity-independent (\cref{sec:formal}).

\item \textbf{Operation-bounded credential model.} The OpenID Foundation's
  execution-count proposal~\cite{openid2025} is grounded in release consistency
  theory~\cite{sorin2020}, demonstrating formal equivalence to acquire/release
  synchronisation primitives at coherence boundaries (\cref{sec:arch}).

\item \textbf{Reproducible evaluation.} Four strategies evaluated across three
  scenarios with multi-run statistical aggregation (10 runs, seeds 0--9):
  $120\times$ reduction in unauthorised ops for execution-count vs.\ \textsc{TTL}
  in the \textsc{CRM} scenario; $184\times$ in the anomaly scenario
  (\cref{sec:eval}).
\end{enumerate}

\section{Theoretical Foundations}
\label{sec:foundations}

\subsection{Cache Coherence Essentials}

Sorin, Hill, and Wood~\cite{sorin2020} define coherence as the requirement that
reads to a memory location return the most recently written value and that writes
serialise.
The \MESI{} protocol---four stable states (Modified, Exclusive, Shared,
Invalid) with well-defined transitions---is the canonical implementation.
Between stable states, transient states~\cite[Ch.6, \S6.4.1]{sorin2020}
(notated $XYZ$: ``was $X$, transitioning to $Y$, awaiting event $Z$'') model
in-flight operations; their duration is precisely the damage window in my
authorisation analogy.

Coherence strategies bifurcate into two classes~\cite[Ch.2, \S2.3]{sorin2020}:
\textit{consistency-agnostic} (\textsc{SWMR}-enforcing, synchronous
invalidation) and \textit{consistency-directed} (relaxed, bounded-staleness).
\textsc{GPU} architectures~\cite{singh2013,alsop2016} adopt the latter,
implementing temporal coherence (lease-based self-invalidation) and release
consistency (synchronisation-point coherence).
The authorisation strategies evaluated in \cref{sec:eval} map onto these classes
exactly.

\subsection{Failure Containment}

Reason's Swiss Cheese Model~\cite{reason2000} holds that failure occurs when
independent defensive layers share aligned failure modes.
Vijayaraghavan et al.~\cite{vijayaraghavan2026} apply this to multi-agent
reliability, demonstrating that cascaded critique layers with orthogonal failure
modes catch 92.1\% of errors.
The alignment problem for \textsc{TTL}-based authorisation is structurally
identical: the failure mode (temporal window) aligns precisely with the threat
vector (operational velocity).
Operation-count bounds are orthogonal to velocity---they create the barrier the
Swiss Cheese model demands.

\subsection{Information-Theoretic Framing}

Shannon~\cite{shannon1948} established that reliable communication over a noisy
channel requires redundancy at the cost of effective bandwidth.
The authorisation channel between authority and agents is noisy (delayed, lossy,
potentially Byzantine).
Coherence enforcement overhead---re-validation checks, heartbeat signals---is
the authorisation-layer equivalent of error-correcting codes: a throughput cost
extracted in exchange for revocation reliability.
The \RCC{} overhead formula
$\mathrm{Overhead}_{\RCC{}} = \Delta_{\mathrm{revalidation}} / n$
(\cref{sec:eval:cost}) makes this correspondence numerically precise.

\section{Formal Model}
\label{sec:formal}

\subsection{Capability Coherence System}

\begin{definition}[Capability Coherence System]
\label{def:ccs}
A \emph{Capability Coherence System} (\CCS{}) is a tuple
$\langle A, C, \Sigma, \delta, \alpha, \B \rangle$ where:
\begin{itemize}[noitemsep]
  \item $A = \{a_1, \ldots, a_m\}$ is a finite set of agents;
  \item $C = \{c_1, \ldots, c_k\}$ is a finite set of capabilities
        (analogous to memory locations);
  \item $\Sigma = \{M, E, S, I\}$ is the set of stable authorisation states;
  \item $\delta: \Sigma \times \E \to \Sigma$ is the state transition function
        over events $\E = \{\mathit{grant}, \mathit{revoke}, \mathit{delegate},$
        $\mathit{introspect}, \mathit{exhaust}, \mathit{expire}\}$;
  \item $\alpha: A \times C \to \Sigma$ maps each agent-capability pair to its
        current authorisation state;
  \item $\B: \Sigma \to 2^{\mathcal{O}}$ is the permitted operations function
        over operations $\mathcal{O}$.
\end{itemize}
\end{definition}

Capability validity maps to cache line validity: $s \in \{M, E, S\}$ is a cached
block with valid data; $I$ is an invalidated line---no read or write proceeds
without a coherence fill from the authority.
Formally, $\B(I) = \emptyset$, mirroring the hardware constraint precisely.

\paragraph{SWMR Adaptation.}
For any capability $c$ at logical time $t$:
\[
  \bigl|\{a \in A : \alpha_t(a, c) = M\}\bigr| \leq 1.
\]
At most one agent holds delegation-capable state simultaneously; multiple agents
may occupy $S$ concurrently.
This mirrors the Single-Writer-Multiple-Reader invariant
of~\cite[Ch.2, \S2.4]{sorin2020}.

\begin{definition}[Authorisation State Machine]
\label{def:asm}
The authorisation state machine
$\M = (\Sigma \cup \Sigma_T, \E, \delta)$ extends \cref{def:ccs} with transient
states $\Sigma_T = \{EIA, SIA, MIC, ISG, IED\}$ following the $XYZ$ notation
of~\cite[Ch.6]{sorin2020}.
\end{definition}

Valid transitions are enumerated in \cref{tab:transitions}.

\begin{table}[h]
\centering
\caption{Valid transitions of the authorisation state machine.}
\label{tab:transitions}
\small
\begin{tabular}{@{}llll@{}}
\toprule
\textbf{From} & \textbf{Event} & \textbf{To} & \textbf{Semantics} \\
\midrule
$I$ & \textit{grant\_shared}    & $S$   & Capability issued (role-based) \\
$I$ & \textit{grant\_exclusive} & $E$   & JIT credential issued \\
$E$ & \textit{delegate}         & $M$   & Agent sub-delegates \\
$E$ & \textit{revoke}           & $EIA$ & Revocation in-flight, ACK pending \\
$EIA$ & \textit{ack}            & $I$   & Revocation confirmed \\
$S$ & \textit{revoke}           & $SIA$ & Revocation in-flight (shared) \\
$M$ & \textit{revoke\_cascade}  & $MIC$ & Cascade revocation, sub-agents pending \\
$MIC$ & \textit{all\_acks}      & $I$   & All delegees confirmed \\
$S$ & \textit{exhaust}          & $I$   & Operation count depleted (\RCC{}) \\
$S$ & \textit{expire}           & $I$   & \textsc{TTL} expired (temporal coherence) \\
\bottomrule
\end{tabular}
\end{table}

Invalid transitions---$I \to M$ (delegation without prior grant), $S \to M$
(bypass of exclusivity upgrade)---are rejected by~$\delta$.
The transient states are not cosmetic: the duration of $EIA$ and $MIC$ is
precisely the interval over which unauthorised operations accumulate.
Formal analyses that collapse transient states into atomic transitions
systematically underestimate the damage window.

\subsection{State Mapping Function}

\begin{definition}[State Mapping]
\label{def:phi}
Let $\Sigma_{\MESI{}} = \{M_{hw}, E_{hw}, S_{hw}, I_{hw}\}$~\cite{sorin2020}
and $\Sigma_{\mathrm{auth}} = \{M, E, S, I\}$.
The mapping $\dphi: \Sigma_{\MESI{}} \to \Sigma_{\mathrm{auth}}$ is:
\begin{align*}
  \dphi(M_{hw}) &= M \quad \text{(write-capable $\mapsto$ delegate-capable)} \\
  \dphi(E_{hw}) &= E \quad \text{(exclusive read $\mapsto$ sole credential holder)} \\
  \dphi(S_{hw}) &= S \quad \text{(shared read $\mapsto$ role-based pooled access)} \\
  \dphi(I_{hw}) &= I \quad \text{(invalid $\mapsto$ revoked)}
\end{align*}
\end{definition}

\begin{figure}[t]
  \centering
  \includegraphics[width=\textwidth]{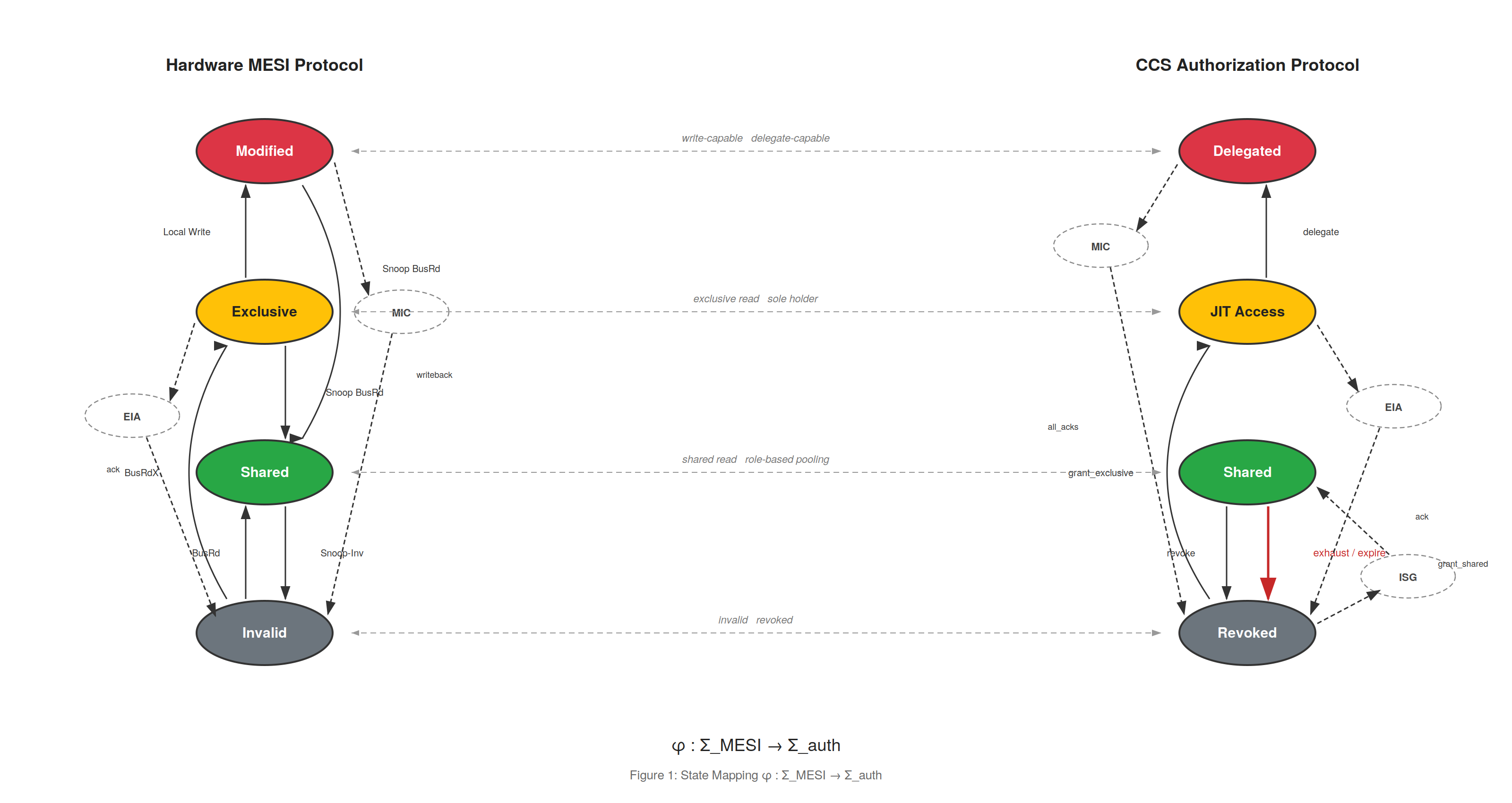}
  \caption{State mapping $\dphi: \Sigma_{\MESI{}} \to \Sigma_{\mathrm{auth}}$.
    Left: hardware \MESI{} state machine with transition triggers
    (\textit{BusRd}, \textit{BusRdX}, \textit{Snoop-Inv}).
    Right: authorisation state machine with corresponding triggers
    (\textit{grant}, \textit{delegate}, \textit{revoke}).
    Transient states as dashed nodes ($EIA$, $MIC$, $ISG$).
    Event correspondences labelled on connecting arrows.}
  \label{fig:phi}
\end{figure}

\begin{claim}[Structural Equivalence]
\label{clm:equiv}
The mapping $\dphi$ preserves transition structure: for every valid transition
$(s_1, e_{hw}, s_2)$ in the hardware \MESI{} protocol, there exists a
corresponding valid transition $(\dphi(s_1), e_{\mathrm{auth}}, \dphi(s_2))$
in $\M$, where $e_{\mathrm{auth}}$ is the authorisation-domain event
corresponding to hardware event $e_{hw}$.
\end{claim}

\begin{proof}[Proof sketch]
Correspondence verified by exhaustive enumeration against~\cite[Table~6.3]{sorin2020}.
Hardware \textit{BusRd} maps to \textit{grant\_shared}; \textit{BusRdX} to
\textit{grant\_exclusive}; \textit{Snoop-Invalidate} to \textit{revoke};
\textit{Write-back} to \textit{delegate} (scope attenuation).
\textsc{SWMR} is preserved by construction (\cref{def:ccs}).
\end{proof}

I use \textit{structural equivalence} rather than \textit{isomorphism} because
the authorisation domain introduces transitions absent from hardware \MESI{}---
notably \textit{exhaust} and \textit{expire}---and imposes the additional
constraint $\mathit{child.scope} \subseteq \mathit{parent.scope}$.
The hardware model is a strict subset of the authorisation model, which is
architecturally convenient: known hardware bounds carry over as lower-bound
estimates for authorisation.

\paragraph{Scope of equivalence.}
\CCS{} establishes an \emph{operational} equivalence under bounded-staleness
semantics, not a semantic isomorphism.
Hardware \MESI{} guarantees memory visibility through physical interconnect
properties; \CCS{} achieves analogous authorisation visibility through
protocol-level strategy constraints.
Domain-specific extensions---trust scoring, delegation \textsc{DAG}s, scope
attenuation---extend the authorisation model beyond the hardware domain.
\CCS{} is a coherence-inspired authorisation protocol; it is not a full \MESI{}
instance.

\subsection{Damage Bound Functions}

\begin{definition}[Velocity Vulnerability]
\label{def:vv}
For agent velocity $v$~(ops/tick) and credential \textsc{TTL} $\Delta t$~(ticks):
\[
  V_v(v, \Delta t) = v \cdot \Delta t.
\]
\end{definition}

The significance is not the arithmetic---which is elementary---but that agent
velocity emerges as a first-class security parameter: a dimension absent from
human-centric \textsc{IAM} literature.

\begin{definition}[Damage Bound]
\label{def:dbound}
For revocation strategy $\pi$, the \emph{damage bound}~$D_\pi$ is the maximum
unauthorised operations executable after the authority initiates revocation:

\begin{center}
\footnotesize
\setlength{\tabcolsep}{4pt}
\begin{tabular}{@{}p{4.0cm}p{5.8cm}p{5.2cm}@{}}
\toprule
\textbf{Strategy} $\pi$ & $D_\pi$ & \textbf{Bound type} \\
\midrule
Eager (consistency-agnostic) & $D_{\mathrm{eager}} \leq v \cdot \Delta_{\mathrm{network}}$ & Time-bounded, velocity-dependent \\[2pt]
Lazy (check-on-use) & $D_{\mathrm{lazy}} \leq v \cdot (\Delta_{\mathrm{network}} + \Delta_{\mathrm{check}})$ & Time-bounded, velocity-dependent \\[2pt]
Lease (temporal coherence) & $D_{\mathrm{lease}} \leq v \cdot \mathrm{TTL}$ & Time-bounded, velocity-dependent \\[2pt]
Exec-count (\RCC{}) & $D_{\mathrm{rcc}} \leq n$ & \textbf{Operation-bounded, velocity-independent} \\
\bottomrule
\end{tabular}
\end{center}
\end{definition}

\begin{theorem}[\RCC{} Safety Bound]
\label{thm:rcc}
In a \CCS{} with execution-count strategy and budget~$n$, the maximum
unauthorised operations per capability after revocation initiation is bounded
by~$n$, independent of agent velocity~$v$.
\end{theorem}

\begin{proof}[Proof sketch]
Let $t_r$ be the authority's revocation time for capability~$c$ held by agent~$a$.
Let $k$ be operations $a$ has executed against~$c$ since the most recent acquire.
At budget exhaustion ($k = n$, the release point), $a$ must contact the authority
for a fresh credential (acquire).
The authority, having revoked~$c$ at $t_r$, denies the acquire.
Maximum unauthorised operations after~$t_r$ are $n - k \leq n$.
Critically, the bound is on \emph{operations}, not time: a faster agent exhausts~$n$
sooner but cannot exceed~$n$.
The proof holds regardless of~$v$.
\end{proof}

\begin{remark}[\cref{thm:rcc} scope]
The bound $D_{\mathrm{rcc}} \leq n$ applies \emph{per capability}.
In delegation chains where multiple agents hold distinct capabilities with
independent budgets, aggregate unauthorised operations across the cascade may
vary with stochastic scheduling---the banking scenario ($\sigma = 13.0$ for
exec-count, \cref{sec:eval:banking}) is a concrete illustration---but each
individual capability respects the deterministic bound.
Zero bound violations across all 120 experimental runs
($4 \times 10 \times 3$ configurations) confirm this without exception.
\end{remark}

\begin{remark}[Distinction from rate limiting]
A rate limiter caps throughput; \RCC{} \emph{invalidates the credential itself}
at the operation boundary.
Under rate limiting, a revoked agent is throttled but retains a credential it
believes valid.
Under \RCC{}, the agent discovers revocation at the synchronisation boundary
and transitions $\alpha(a, c) \to I$.
The agent cannot resume without a fresh grant---this is acquire semantics.
The behavioural distinction is observable in \cref{tab:crm}: \RCC{} is the only
strategy that forces re-acquisition ($1.0 \pm 0$ revalidations) rather than
silently expiring state.
\end{remark}

\begin{corollary}
\label{cor:eager}
For eager strategy: $D_{\mathrm{eager}} = v \cdot \Delta_{\mathrm{network}}$.
At $v = 100$, $\Delta_{\mathrm{network}} = 5$: $D_{\mathrm{eager}} = 500$.
Eager invalidation approaches zero unauthorised operations only when
$\Delta_{\mathrm{network}} < 1/v$---an unattainable condition for realistic
high-velocity deployments.
\end{corollary}

\section{Related Work}
\label{sec:related}

\subsection{Multi-Agent Failure Modes}

Building on ReAct~\cite{yao2023} and AutoGen~\cite{wu2023}---which document the
transition from singular \textsc{LLM}s to compound agent systems---neither
framework treats the authorisation state of individual agents, a gap with direct
security consequences.
Chain-of-Thought~\cite{wei2022} improves reasoning reliability without touching
authorisation state.
Cemri et al.~\cite{cemri2025} provide the most comprehensive failure taxonomy to
date: \textsc{MAST} (1,600+ traces, seven frameworks) identifies fourteen failure
modes across three categories.
Their ``inter-agent misalignment''---agents diverging from shared
expectations---maps directly to the coherence violation I formalise: an agent in
perceived state~$S$ while the authority has moved it to~$I$.
This is stale read semantics at the authorisation layer.

\subsection{Organisational Reliability Models}

Building on~\cite{vijayaraghavan2026}---the ``Team of Rivals'' architecture---I
extend their organisational metaphor from semantic coherence
(Generator/Critic separation) to security coherence
(Execution/Authorisation separation).
Huang et al.~\cite{huang2024} provide empirical grounding: hierarchical oversight
yields 5.5\% performance degradation versus 23.7\% for flat topologies, and their
Inspector mechanism recovers 96.4\% of faulty agent errors.
The Authority Service in \CCS{} plays the Inspector role at the authorisation
layer.
The 2--10\% \RCC{} overhead (\cref{sec:eval:cost}) is commensurate with the
Huang et al.\ benchmark.

\subsection{Trust and Security Governance}

Raza et al.~\cite{raza2025} adapt \textsc{TRiSM} for agentic multi-agent
systems, defining monitoring metrics (Component Synergy Score, Tool Utilisation
Efficacy).
Their framework specifies \emph{what} to monitor; \CCS{} delineates the
\emph{enforcement mechanism}: a declining trust score triggers a concrete state
transition $S \to I$ in the state machine.
The \textsc{CSA} framework~\cite{csa2025} articulates ``Continuous
Authorisation''---privileges evaluated continuously rather than granted once at
session initiation.
My architecture instantiates this as a protocol layer through which trust signals
become coherence actions.

\subsection{Delegation and Authorisation Standards}

Schwenkschuster et al.~\cite{schwenkschuster2024} formalise cross-domain identity
chaining; South et al.~\cite{south2025} provide cryptographic delegation
primitives.
Both address chain \emph{establishment}.
I address chain \emph{teardown}---the coherent propagation of revocation across
delegation \textsc{DAG}s.
Establishment is the easy half of the problem.

For scope attenuation, two approaches have emerged: online
(\textsc{RFC}~8693 Token Exchange~\cite{campbell2020}) and offline
(Biscuits~\cite{biscuit2024}, Macaroons~\cite{birgisson2014}).
Both are \emph{stateless} scope management mechanisms.
\CCS{} provides the complementary \emph{stateful} coherence layer: even when
delegation is correctly authenticated and attenuated, temporal consistency of
revocation remains---as~\cite{openid2025} acknowledges---``largely unsolved.''
The two mechanisms are composable, not competing.

\section{Authorisation as a Coherence Problem}
\label{sec:coherence}

In centralised systems, authorisation is atomic: a single database check yields
a binary, current result.
In distributed agent systems, authorisation is \emph{cached}---every token and
signed assertion is a cached copy of a permission that existed at~$t_0$.
The distributed authorisation literature has addressed the performance
implications of this caching but has largely neglected the coherence implications.

The difficulty of securing agentic systems stems from treating the stale-token
problem as a latency problem when it is structurally a coherence problem.
In Sorin et al.'s taxonomy~\cite{sorin2020}, coherence protocols manage the
propagation of writes (revocations) to all readers (agents).
The ``write'' is a revocation event; the ``readers'' are agents holding cached
credentials.
The coherence strategy determines how quickly and at what cost the write propagates.

\subsection{Consistency-Agnostic (Eager) Revocation}

This class enforces \textsc{SWMR} synchronously: the authority invalidates all
held tokens before acknowledging the revocation.
In hardware terms, this is bus-based snooping.
The fragility is structural, not incidental: a single offline agent causes the
process to hang or fail open.
Cemri et al.'s \textsc{MAST} taxonomy~\cite{cemri2025} documents ``system design
issues'' as a primary failure category; synchronous revocation introduces
precisely this mode at the authorisation layer.
\Cref{cor:eager} adds a further constraint: eager strategies do not achieve zero
unauthorised operations when $\Delta_{\mathrm{network}} > 0$---at $v = 100$,
the 5-tick propagation window admits 500 unauthorised operations regardless of
broadcast aggressiveness.

\subsection{Consistency-Directed (Relaxed) Revocation}

The system accepts asynchronous revocation propagation under bounded staleness.
The question---where the interesting design space lies---is \emph{what bounds the
staleness}.
Temporal coherence (\textsc{TTL}) bounds staleness by time, creating the
$O(v \cdot \mathrm{TTL})$ vulnerability of \cref{def:vv}.
Release consistency (execution-count) bounds staleness by operations, eliminating
velocity dependence by \cref{thm:rcc}.
The qualitative departure between these two bound types is the central empirical
claim of this paper.

\section{The MESI Mapping}
\label{sec:mesi}

I instantiate $\dphi$ from \cref{def:phi} on the agent authorisation lifecycle
(\cref{fig:phi}).
Each stable state carries a concrete operational semantics:

\paragraph{Modified $\to$ Delegated Authority.}
The agent holds delegation rights---in Biscuit/Macaroon terms~\cite{biscuit2024,birgisson2014},
it can append attenuation caveats.
Revocation from~$M$ requires cascade notification:
$M \xrightarrow{\mathit{revoke\_cascade}} MIC$.
The in-flight $MIC$ interval is where cascade damage accumulates.

\paragraph{Exclusive $\to$ JIT Access.}
Single-use credential; no sharing.
Maps to \textsc{CSA} ``Just-in-Time Access''~\cite{csa2025}.
The credential exists for exactly the task duration---the shortest possible
$D_\pi$ window for time-bounded strategies.

\paragraph{Shared $\to$ Role-Based Pooling.}
Multiple agents hold read-only copies simultaneously
(e.g., ``Read S3 Bucket'').
Default state for most OAuth access tokens.
Staleness risk is multiplicative: all copies are stale until individually
invalidated.

\paragraph{Invalid $\to$ Revoked.}
$\B(I) = \emptyset$.
The agent cannot proceed without a coherence transaction.

The transient state $EIA$---revocation sent, \textsc{ACK} pending---is where the
bulk of damage occurs in eager strategy deployments.
Formal analyses that treat revocation as instantaneous systematically
underestimate the damage window by omitting this interval.

\section{Architecture}
\label{sec:arch}

\subsection{Execution-Count Bounds as Release Consistency}

In \textsc{GPU} release consistency~\cite[Ch.10, \S10.1.4]{sorin2020,alsop2016},
coherence is enforced at synchronisation points: \textit{acquire} loads fresh
state, \textit{release} commits and exposes writes.
Between these boundaries, the processor operates on its local cache without
coherence traffic.

An operation budget~$n$ is issued with each credential.
Each operation decrements a counter.
At $k = n$---budget exhaustion, the release point---the agent contacts the
authority for a fresh credential (acquire).
If the capability was revoked during the $n$-operation critical section, the
acquire is denied.
By \cref{thm:rcc}, the damage is bounded at~$n$ regardless of~$v$.

The OpenID Foundation~\cite{openid2025} independently proposes ``credentials
constrained by execution counts.''
I demonstrate this is formally equivalent to release/acquire cycles: each
$n$-operation block is a critical section bounded by acquire (credential refresh)
and release (budget exhaustion).
To my knowledge, this equivalence has not been stated in the authorisation
literature; it is the key theoretical contribution enabling transfer of hardware
coherence bounds to the authorisation domain.

\subsection{Authority Service}

The Authority Service combines a Policy Decision Point (\textsc{PDP})~\cite{hu2014}
with a coherence directory controller:

\begin{enumerate}[leftmargin=*, noitemsep]
\item \textbf{Capability Registry.}
  Ground truth of all active capabilities, delegation \textsc{DAG}s, and scope
  attenuation chains.
  Enforces $\mathit{child.scope} \subseteq \mathit{parent.scope}$ as a hard
  invariant.
  (Under network partition, the registry may diverge from agents' cached views;
  see \cref{sec:limitations}.)

\item \textbf{Trust Scorer.}
  Implements \textsc{CSA} continuous authorisation~\cite{csa2025} and Raza
  et al.'s \textsc{TRiSM} monitoring~\cite{raza2025}.
  Behavioural analytics produce dynamic trust scores; score below threshold
  $\tau < 0.4$ triggers automatic revocation:
  $\alpha(a, c) \to I \ \forall\, c \in C_a$.
  The threshold~$\tau$ is treated as scenario-fixed; sensitivity analysis across
  $\tau \in [0.2, 0.6]$ is deferred.

\item \textbf{Revocation Broadcaster.}
  Snooping-style broadcast for $\leq 25$ agents; directory-based unicast for
  larger populations~\cite[Ch.6, \S6.4.3]{sorin2020}.
  The snooping-to-directory crossover was not empirically validated in this
  study---flagged as a pending calibration item.
\end{enumerate}

\subsection{Heterogeneous Coherence}
\label{sec:arch:hetero}

Building on compound consistency models in \textsc{CPU}+\textsc{GPU}
systems~\cite[Ch.10, \S10.2]{sorin2020}, different agent classes operate under
different strategies.
Strategy assignment maps directly from the consistency class taxonomy:
high-integrity operations (financial transactions) require
consistency-agnostic semantics (eager); high-throughput operations
(bulk \textsc{CRM} updates) tolerate consistency-directed relaxation
(lease, exec-count).
Each device operates under its own consistency model; the Authority Service
provides cross-domain ordering---the global coherence controller.

\begin{table}[h]
\centering
\caption{Heterogeneous coherence strategy assignment by agent context.}
\label{tab:hetero}
\small
\begin{tabular}{@{}llll@{}}
\toprule
\textbf{Agent Context} & \textbf{Coherence Strategy} & \textbf{Damage Bound} & \textbf{Hardware Analogy} \\
\midrule
Financial ops     & Eager (consistency-agnostic) & $\leq v \cdot \Delta_{\mathrm{net}}$   & CPU SWMR \\
CRM / Bulk sync   & Lease-based (temporal)       & $\leq v \cdot \mathrm{TTL}$            & GPU Temporal \\
Analytics         & Lazy (check-on-use)          & $\leq v \cdot \Delta_{\mathrm{check}}$ & Weak ordering \\
High-velocity API & Exec-count (\RCC{})          & $\leq n$~ops                           & GPU Release Consistency \\
\bottomrule
\end{tabular}
\end{table}

\begin{figure}[t]
  \centering
  \includegraphics[width=\textwidth]{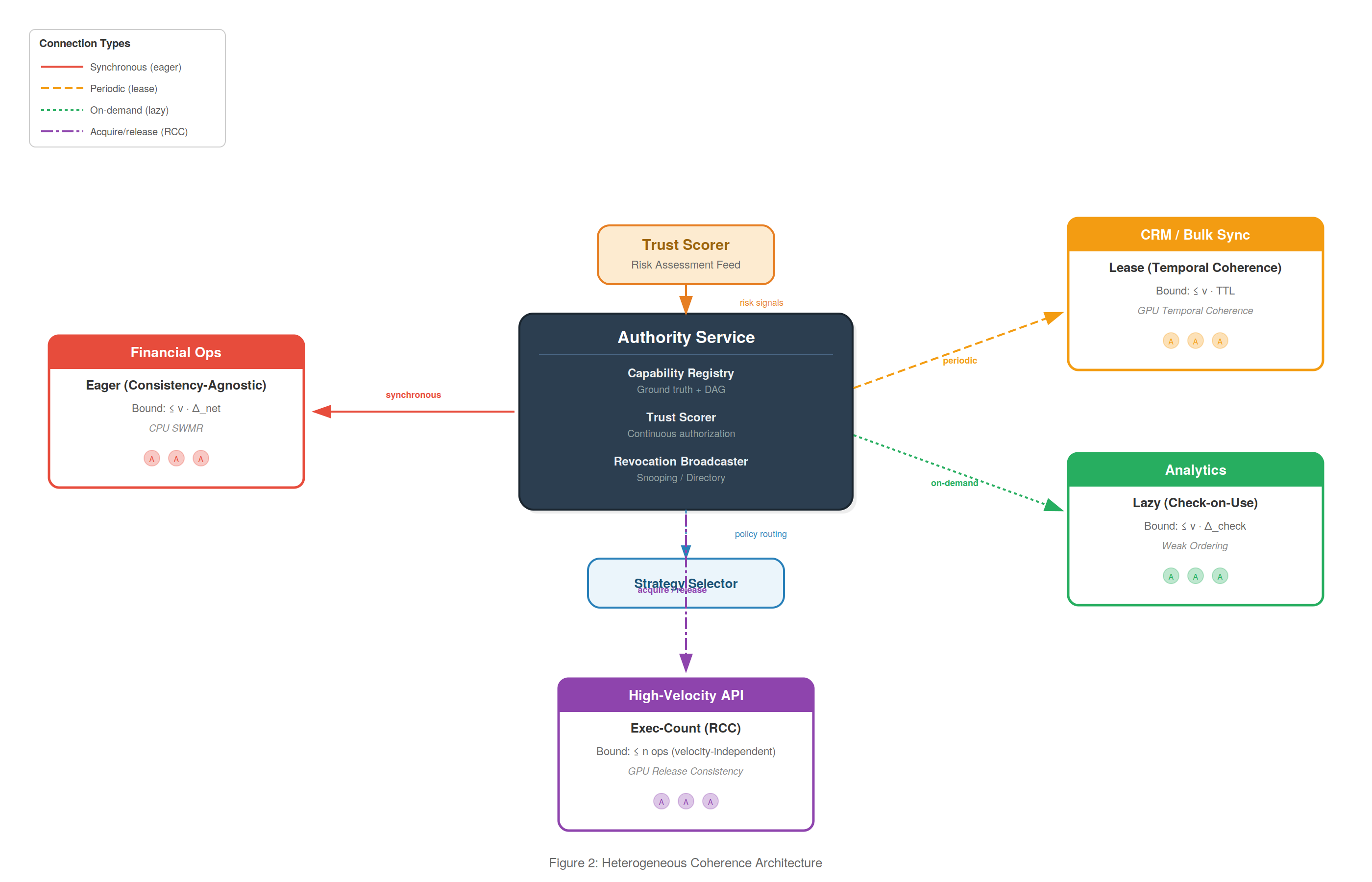}
  \caption{Heterogeneous Coherence Architecture.
    Authority Service (central) connected to four agent clusters, each labelled
    with coherence strategy.
    Message types: synchronous (eager, solid), periodic (lease, dashed),
    on-demand (lazy, dotted), acquire/release (exec-count, double).
    Trust Scorer feeds risk assessments into the strategy selector.}
  \label{fig:arch}
\end{figure}

\section{Evaluation}
\label{sec:eval}

\subsection{Simulation Methodology}
\label{sec:eval:method}

A tick-based discrete event simulator was implemented in Python~3.11.
Each tick is an indivisible scheduling unit; all operations within a tick are
logically simultaneous.
The tick abstraction isolates coherence dynamics from deployment-specific
latencies.
To map to a specific deployment, tick counts are multiplied by the deployment's
scheduling interval (e.g., one tick $= 10$~ms implies 5-tick latency $= 50$~ms).

\paragraph{Methodological caveat.}
The tick simulator introduces a minor temporal aliasing artefact when operations
cluster at tick boundaries.
In probabilistic scheduling scenarios (Banking, Anomaly), operations at $p = 0.5$
and $p = 0.7$ are resolved by independent Bernoulli draws per agent per tick.
For the deterministic \textsc{CRM} scenario ($v = 100$~ops/tick constant), tick
alignment eliminates this concern entirely and produces $\sigma = 0$ across all
seeds---which is correct but deserves explicit acknowledgment rather than being
reported as a surprising result.

\begin{algorithm}[t]
\caption{Simulation Engine}
\label{alg:sim}
\begin{algorithmic}[1]
\Require scenario config $S$, strategy $\pi$, max\_ticks $T$
\Ensure metrics $M$
\State Initialise \texttt{AuthorityService}, \texttt{AgentSet} from $S$
\State Initialise \texttt{MetricsCollector}
\For{$t = 1$ \textbf{to} $T$}
  \ForAll{agent $a \in \texttt{AgentSet}$}   \Comment{Phase 1: Agent operations}
    \State $\mathit{attempts} \gets \texttt{determine\_action\_count}(a, S, t)$
    \ForAll{attempt}
      \State $\mathit{cap} \gets \texttt{random\_capability}(a)$
      \State $\texttt{result} \gets a.\texttt{attempt\_operation}(\mathit{cap})$
      \If{$\mathit{cap}.\texttt{actually\_revoked\_at\_authority}()$}
        \State $M.\texttt{unauthorized\_count} \mathrel{+}= 1$
      \EndIf
      \If{$\pi = \RCC{}$ \textbf{and} $a.\texttt{ops\_remaining} = 0$}
        \State $a.\texttt{request\_revalidation}()$ \Comment{acquire}
      \EndIf
    \EndFor
  \EndFor
  \ForAll{pending revocation $r$}            \Comment{Phase 2: Authority processing}
    \If{$\pi = \textsc{Eager}$}   \State \texttt{broadcast\_and\_await\_ack}$(r)$
    \ElsIf{$\pi = \textsc{Lazy}$} \State \texttt{mark\_revoked}$(r)$
    \ElsIf{$\pi = \textsc{Lease}$} \State \texttt{mark\_revoked}$(r)$
    \ElsIf{$\pi = \RCC{}$}        \State \texttt{mark\_revoked}$(r)$
    \EndIf
  \EndFor
  \State \texttt{deliver\_messages}$(\mathrm{latency} = S.\texttt{latency\_ticks})$  \Comment{Phase 3: Network delivery}
  \ForAll{agent $a$}                         \Comment{Phase 4: Anomaly detection}
    \If{\texttt{trust\_scorer.check\_anomaly}$(a)$}
      \State \texttt{revoke\_all\_capabilities}$(a)$
    \EndIf
  \EndFor
  \ForAll{agent $a$ with transient state}    \Comment{Phase 5: Resolve transients}
    \If{\texttt{ack\_received}$(a)$}
      \State $a.\texttt{state} \gets \texttt{stable\_target}(a.\texttt{transient\_state})$
    \EndIf
  \EndFor
  \State $M.\texttt{record\_tick}(t, \texttt{AgentSet})$  \Comment{Phase 6: Metrics}
\EndFor
\State \Return $M.\texttt{aggregate}()$
\end{algorithmic}
\end{algorithm}

\paragraph{Multi-run aggregation.}
Each strategy-scenario configuration is executed 10 times with deterministic
seeds 0--9.
Population mean and population standard deviation~$\sigma$ (not sample~$s$) are
reported, since the 10 runs constitute the complete experimental population
under the specified seed range.
All scenario configurations are published as \textsc{YAML} files alongside the
source code at \url{https://github.com/hipvlady/prizm}.

\subsection{Scenario Configurations}
\label{sec:eval:scenarios}

\begin{table}[h]
\centering
\caption{Simulation parameters across three business-contextualised scenarios.}
\label{tab:params}
\small
\begin{tabular}{@{}llll@{}}
\toprule
\textbf{Parameter} & \textbf{Scenario 1 (Banking)} & \textbf{Scenario 2 (CRM)} & \textbf{Scenario 3 (Anomaly)} \\
\midrule
Agents             & 10               & 1               & 5 \\
Delegation depth   & 3                & 1               & 1 \\
Action model       & Prob.\ ($p=0.5$) & Det.\ (100/tick) & Prob.\ ($p=0.7$) \\
Seeds              & 0--9 (10 runs)   & 0--9 (10 runs)  & 0--9 (10 runs) \\
Network latency    & 10 ticks         & 5 ticks         & 10 ticks \\
Revocation trigger & Tick 100         & Tick 0          & Auto (trust score) \\
TTL (lease, ticks) & 120              & 60              & 3,000 \\
Exec-count $n$     & 60               & 50              & 100 \\
Lazy check interval& 40 ticks         & 23 ticks        & 10 ticks \\
Anomaly burst      & N/A              & N/A             & 12 ops/tick \\
Trust threshold $\tau$ & 0.8          & 0.8             & 0.4 \\
Trust decay        & 0.3              & 0.3             & 0.5 \\
Duration (ticks)   & 200              & 120             & 300 \\
\bottomrule
\end{tabular}
\end{table}

\subsection{Scenario 1 --- Banking Cascade Revocation}
\label{sec:eval:banking}

A payment processing agent is compromised at tick~100.
A three-level delegation chain (User $\to$ A $\to$ B $\to$ C) requires cascade
revocation.
Probabilistic scheduling ($p = 0.5$) produces genuine cross-seed variance.

\begin{table}[h]
\centering
\caption{Banking cascade results (mean $\pm$ $\sigma$, 10 runs, seeds 0--9).}
\label{tab:banking}
\small
\begin{tabular}{@{}lllll@{}}
\toprule
\textbf{Metric} & \textbf{Eager} & \textbf{Lease (120-tick)} & \textbf{Lazy (40-tick)} & \textbf{RCC ($n=60$)} \\
\midrule
Unauthorised ops     & $14.9 \pm 3.2$  & $29.9 \pm 2.6$ & $35.3 \pm 2.9$ & $32.3 \pm 13.0$ \\
Staleness max (ticks)& $10.0 \pm 0$    & $20.0 \pm 0$   & $23.0 \pm 0$   & $37.0 \pm 21.6$ \\
Messages sent        & $3.0 \pm 0$     & $3.0 \pm 0$    & $3.0 \pm 0$    & $3.0 \pm 0$ \\
Bound violations     & 0               & 0              & 0              & 0 \\
\bottomrule
\end{tabular}
\end{table}

Eager achieves the tightest observed bound ($14.9 \pm 3.2$)---but the non-zero
value deserves direct acknowledgment.
The 10-tick network latency across 10 agents in a 3-level delegation tree admits
operations during propagation, consistent with \cref{cor:eager}.
Eager does not achieve zero; it achieves $v \cdot \Delta_{\mathrm{network}}$
in expectation.

\RCC{} exhibits the highest variance ($\sigma = 13.0$).
This is not a violation of \cref{thm:rcc}: the theorem bounds $D \leq n$
\emph{per capability}, and zero bound violations confirm every individual
capability respected its budget.
The variance arises from an interaction between release/acquire boundaries at
different delegation depths and stochastic action scheduling---depending on the
seed, different agents exhaust their budgets at different points relative to the
cascade propagation wavefront.
The phenomenon is stochastic staleness \emph{within} the bound, not stochastic
violations \emph{of} the bound.

The proximity of all strategies (14.9--35.3 range) reflects the low effective
velocity under probabilistic scheduling.
The \textsc{CRM} scenario (\cref{sec:eval:crm}) is designed to maximally expose
velocity dependence.

\subsection{Scenario 2 --- CRM High-Velocity Agent}
\label{sec:eval:crm}

A sales synchronisation agent operates at 100~ops/tick.
Credential revoked at tick~0---worst case: the entire simulation window is
post-revocation.

\begin{table}[h]
\centering
\caption{CRM high-velocity results (mean $\pm$ $\sigma$, 10 runs, seeds 0--9).}
\label{tab:crm}
\small
\begin{tabular}{@{}lllll@{}}
\toprule
\textbf{Metric} & \textbf{Eager} & \textbf{Lease (60 ticks)} & \textbf{Lazy (23-tick)} & \textbf{RCC ($n=50$)} \\
\midrule
Unauthorised ops     & $500.0 \pm 0$     & $\mathbf{6{,}000.0 \pm 0}$ & $2{,}400.0 \pm 0$ & $\mathbf{50.0 \pm 0}$ \\
Staleness max (ticks)& $5.0 \pm 0$       & $60.0 \pm 0$               & $24.0 \pm 0$      & $0.0 \pm 0$ \\
Revalidation count   & $0.0 \pm 0$       & $0.0 \pm 0$                & $0.0 \pm 0$       & $1.0 \pm 0$ \\
Bound violations     & 0                 & 0                          & 0                 & 0 \\
\bottomrule
\end{tabular}
\end{table}

$\sigma = 0$ across all strategies and all seeds---a direct consequence of
deterministic scheduling ($v = 100$ constant).
The zero variance is a confirmatory result: it proves that damage bounds from
\cref{def:dbound} are \emph{exact}, not stochastic approximations.

\begin{table}[h]
\centering
\caption{Predicted vs.\ observed damage bounds (CRM scenario).}
\label{tab:bounds}
\small
\begin{tabular}{@{}llll@{}}
\toprule
\textbf{Strategy} & \textbf{Predicted $D_\pi$} & \textbf{Observed} & \textbf{Match} \\
\midrule
Eager: $v \cdot \Delta_{\mathrm{net}} = 100 \times 5$ & 500 & $500.0 \pm 0$ & Exact \\
Lazy: $v \cdot (\Delta_{\mathrm{net}} + \Delta_{\mathrm{check}}) = 100 \times (1 + 23)$ & 2,400 & $2{,}400.0 \pm 0$ & Exact \\
Lease: $v \cdot \mathrm{TTL} = 100 \times 60$ & 6,000 & $6{,}000.0 \pm 0$ & Exact \\
\RCC{}: $n$ & 50 & $50.0 \pm 0$ & Exact (\cref{thm:rcc}) \\
\bottomrule
\end{tabular}
\end{table}

The $120\times$ reduction from lease to \RCC{} ($6{,}000 \to 50$) is the central
empirical result (\cref{fig:barchart}).
\textsc{TTL} damage scales $O(v \cdot \mathrm{TTL})$; \RCC{} damage is capped at
$O(n)$ regardless of velocity.
Eager achieves a $12\times$ improvement over lease---but this is still an order
of magnitude above \RCC{}, and it imposes synchronous coupling that introduces
availability risk (\cref{sec:coherence}).

\begin{figure}[t]
  \centering
  \includegraphics[width=0.88\textwidth]{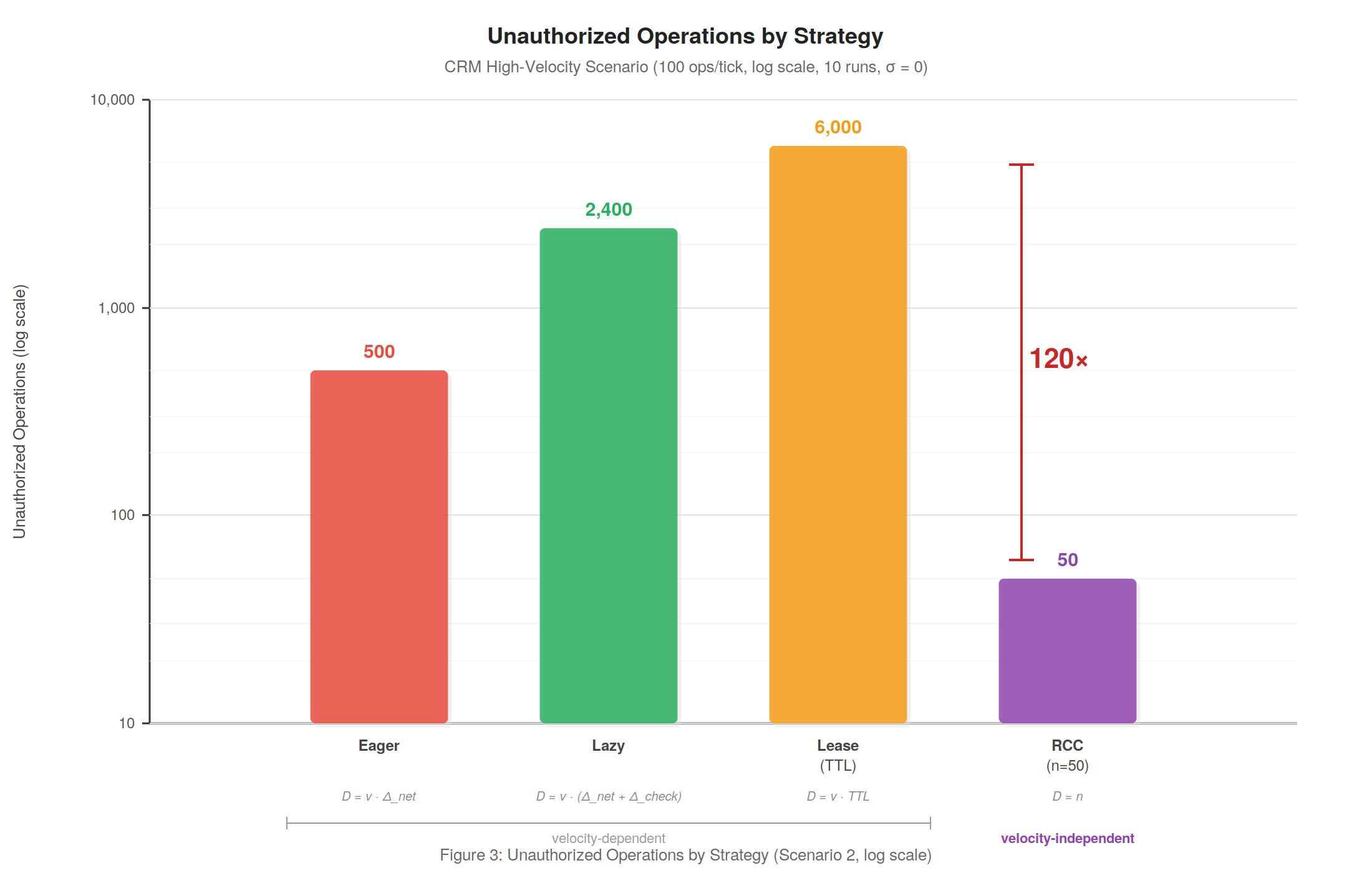}
  \caption{Unauthorised operations by strategy (Scenario 2, log scale).
    \textsc{CRM} high-velocity scenario (100~ops/tick, 10 runs, $\sigma = 0$).
    The $120\times$ gap between Lease and \RCC{} is annotated.
    Y-axis logarithmic; damage bound formulae shown below each bar.}
  \label{fig:barchart}
\end{figure}

\subsection{Scenario 3 --- Anomaly Auto-Revocation}
\label{sec:eval:anomaly}

Agent behaviour bifurcates to burst mode (12~ops/tick) at tick~50.
No explicit revocation is issued; the Trust Scorer detects the anomaly, drops
the trust score below $\tau = 0.4$, and triggers automatic revocation.

\begin{table}[h]
\centering
\caption{Anomaly auto-revocation results (mean $\pm$ $\sigma$, 10 runs, seeds 0--9).}
\label{tab:anomaly}
\small
\begin{tabular}{@{}lllll@{}}
\toprule
\textbf{Metric} & \textbf{Eager} & \textbf{Lazy} & \textbf{Lease (3000-tick)} & \textbf{RCC ($n=100$)} \\
\midrule
Unauthorised ops     & $108.0 \pm 0$  & $10.8 \pm 3.6$ & $\mathbf{2{,}950.8 \pm 3.6}$ & $\mathbf{16.0 \pm 3.7}$ \\
Staleness max (ticks)& $10.0 \pm 0$   & $1.9 \pm 0.3$  & $245.9 \pm 0.3$              & $1.8 \pm 0.4$ \\
Cascade completeness & 1.0            & 1.0            & 0.0                          & 1.0 \\
\bottomrule
\end{tabular}
\end{table}

\RCC{} achieves a $184\times$ reduction against lease \textsc{TTL}
($16.0 \pm 3.7$ vs.\ $2{,}950.8 \pm 3.6$), consistent with the
velocity-independence guarantee of \cref{thm:rcc} under dynamic revocation
triggers.

The counterintuitive result: \textbf{lazy outperforms eager}
($10.8 \pm 3.6$ vs.\ $108.0 \pm 0$).
Eager's synchronous blocking during revocation propagation \emph{delays} the
revocation signal processing while the agent continues operating on its cached
credential.
The Trust Scorer triggers revocation; the lazy agent's next check falls within
its 10-tick check interval---only $1.9 \pm 0.3$ ticks of staleness.
The eager agent, by contrast, accumulates 108 unauthorised operations during
the 10-tick broadcast across 5 agents, compounded by the 12~ops/tick burst
rate during the anomaly window.

This result demonstrates that the four strategies do not form a simple linear
ranking whose ordering is preserved across scenarios.
Strategy effectiveness is topology-dependent.
Lazy check intervals shorter than network propagation latency are advantageous
under trust-triggered revocation; the same check interval would be disastrous
under high-velocity deterministic scheduling (\textsc{CRM}:
Lazy $= 2{,}400$ vs.\ \RCC{} $= 50$).
The \RCC{} result ($16.0 \pm 3.7$) is the reliable backstop: regardless of
which strategy performs better in a specific topology, \RCC{} provides a bounded
ceiling.

The Lease strategy's cascade completeness of~0.0 is a clean result: with a
3,000-tick \textsc{TTL} and a 300-tick simulation, no lease expires within the
simulation window.
The agent never self-invalidates.

\subsection{Cost-Benefit Analysis}
\label{sec:eval:cost}

Theoretical revalidation overhead for \RCC{}:
\[
  \mathrm{Overhead}_{\RCC{}} = \frac{\Delta_{\mathrm{revalidation}}}{n}.
\]
With $\Delta_{\mathrm{revalidation}} = 1$~tick (local re-auth) and $n = 50$:
overhead $= 2.0\%$ of operations require a re-validation round-trip.
At $n = 10$: overhead $= 10.0\%$.
Overhead is inversely proportional to~$n$, instantiating a tunable
security-performance knob.
At $n < 5$, overhead approaches eager-strategy levels, eliminating the
consistency-directed advantage.

\begin{table}[h]
\centering
\caption{Strategy comparison summary.}
\label{tab:summary}
\small
\begin{tabular}{@{}lllll@{}}
\toprule
\textbf{Strategy} & \textbf{Damage Bound} & \textbf{Deterministic?} & \textbf{Vel.-Independent?} & \textbf{Coupling} \\
\midrule
Eager    & $v \cdot \Delta_{\mathrm{net}}$ & Yes & No  & Synchronous \\
Lease    & $v \cdot \mathrm{TTL}$          & Yes & \textbf{No}  & None \\
Lazy     & $v \cdot (\Delta_{\mathrm{net}} + \Delta_{\mathrm{check}})$ & Yes & \textbf{No} & On-demand \\
\RCC{}   & $n$                             & \textbf{Yes} & \textbf{Yes} & Acquire/release \\
\bottomrule
\end{tabular}
\end{table}

\RCC{} is the only strategy whose damage bound is independent of both agent
velocity and network latency while avoiding synchronous coupling.

\section{Discussion}
\label{sec:discussion}

\subsection{Scope and Applicability}

The operational equivalence of \cref{clm:equiv} holds under bounded-staleness
semantics for the core \MESI{} states and transitions.
The authorisation domain's extensions---scope attenuation, trust-based
revocation, hierarchical delegation \textsc{DAG}s---strictly extend the
hardware model.
Hardware bounds function as lower-bound estimates, not tight characterisations.

Huang et al.'s~\cite{huang2024} empirical support---hierarchical oversight
recovers 96.4\% of faulty errors at 5.5\% overhead---grounds the
overhead-for-safety trade inherent in coherence enforcement.
The 2--10\% \RCC{} overhead is the authorisation-layer equivalent.
This alignment is suggestive; I resist overclaiming, since the Huang et al.\
setting (semantic error recovery) and the \CCS{} setting (authorisation state
coherence) are related by analogy, not by formal equivalence.

\subsection{Limitations}
\label{sec:limitations}

\paragraph{Authority centralisation.}
\CCS{} assumes authority consistency strictly stronger than agent
consistency---the reliable directory controller assumption of \textsc{NUMA}
coherence~\cite[Ch.6, \S6.4.3]{sorin2020}.
Under network partition, agents operating on cached credentials remain within
their strategy bounds: exec-count agents are bounded by remaining operations
regardless of authority reachability.
Authority recovery follows directory controller failover: pending revocations
are replayed from a durable event log.
Split-brain scenarios are deferred to future work on replicated authority
services.

\paragraph{Adversarial scheduling.}
The evaluation uses deterministic seeds under fair scheduling.
\RCC{} bounds the \emph{quantity} of unauthorised operations, not their
\emph{semantic impact}.
A strategic adversary could select which $n$ operations to execute for maximum
damage---a capability with no hardware analogy but entirely plausible for a
compromised agent acting on learned objectives.
Adversarial operation scheduling within the execution budget is an open problem.

\paragraph{Simulation scale.}
Three scenarios, 1--10 agents, delegation depths 1--3.
Production systems introduce Byzantine failures, network partitions, and
heterogeneous identity providers.
Stress testing at 100+ agents is future work; the snooping-to-directory
transition heuristic requires empirical calibration at that scale.

\paragraph{Tick abstraction.}
Logical ticks abstract wall-clock time.
The relative comparisons ($120\times$, $184\times$ reductions) hold regardless
of tick-to-time mapping; absolute timing claims require a deployment-specific
mapping.

\paragraph{Formal verification.}
\Cref{clm:equiv} is verified by exhaustive transition enumeration, not by
mechanised proof.
A \textsc{TLA+} specification would enable model checking of safety (no
unauthorised operations after acquire) and liveness (all agents eventually reach
consistent state).
This is the primary formalisation avenue for future work.

\paragraph{Context contamination.}
Credential revocation does not revoke \emph{knowledge}.
An agent that read 50 records retains that context.
Revocation is necessary but not sufficient for full containment.

\paragraph{Emerging standards.}
Integration with \textsc{OIDC-A}~\cite{nagabhushanaradhya2025},
\textsc{SCIM} Agentic Schema~\cite{wahl2024}, and identity
chaining~\cite{schwenkschuster2024} requires protocol-level formalisation as
those standards mature.

\section{Conclusion}
\label{sec:conclusion}

Authorisation revocation in multi-agent delegation chains is operationally
equivalent to cache coherence in shared-memory multiprocessors under
bounded-staleness semantics.
The formal model (\cref{sec:formal}) establishes this through explicit state
mapping, a safety theorem, and damage bound analysis.
\textsc{TTL}-based approaches yield velocity-dependent damage scaling as
$O(v \cdot \mathrm{TTL})$.
Execution-count bounds (\RCC{}) yield a velocity-independent, deterministic bound
of $n$ operations per capability (\cref{thm:rcc})---enforcing coherence at
synchronisation boundaries rather than rejecting individual operations.

Tick-based simulation across three scenarios (120 total runs) confirms: predicted
bounds from \cref{def:dbound} match observed values exactly across all
deterministic configurations; \RCC{} achieves $120\times$ reduction versus
\textsc{TTL} in the \textsc{CRM} scenario and $184\times$ in the anomaly
scenario; zero bound violations across all 120 runs confirm the per-capability
scope of \cref{thm:rcc}.
High-velocity agents---any autonomous system executing above approximately
10~ops/second---should not hold time-bounded credentials.
Operation-bounded credentials force periodic re-authorisation at boundaries
defined by $n$, ensuring that revocation is discovered within a bounded number
of operations irrespective of how fast the agent operates.

\bibliographystyle{plain}
\bibliography{refs}

\end{document}